\documentclass[10pt]{article}

\usepackage[letterpaper, margin=1in]{geometry}
\usepackage[letterpaper]{geometry}
\usepackage{algorithm}
\usepackage{algorithmic}
\usepackage{caption}
\usepackage{subcaption}
\usepackage{authblk}
\usepackage{setspace}
\usepackage{ourStyle}
\usepackage[inline]{enumitem}
\usepackage{amsmath,amssymb,graphicx,epstopdf,textpos,rotating}

\usepackage[parfill]{parskip}

\usepackage{fancyhdr}

\fancypagestyle{firststyle}
{
   \fancyhf[L]{\footnotesize In 2017 International Conference on Autonomous Agents \& Multiagent Systems (AAMAS), S\~ao Paulo. \cite{de2017understanding}}
   \fancyfoot{}
}

\title{Understanding Norm Change: An Evolutionary \\ Game-Theoretic Approach (Extended Version)}
\author[S. De et al.]
       {Soham De$^1$, Dana S. Nau$^{1, 2}$, and Michele J. Gelfand$^3$
       \\
       $^1$Department of Computer Science \\ $^2$Institute for Systems Research\\ $^3$Department of Psychology\\
       University of Maryland, College Park, MD, USA\\
       \texttt{\{sohamde, nau\}@cs.umd.edu, mgelfand@umd.edu}
       }

\begin{document}
\maketitle

\begin{abstract}
Human societies around the world interact with each other by developing and maintaining social norms, and it is critically important to understand how such norms emerge and change. In this work, we define an evolutionary game-theoretic model to study how norms change in a society, based on the idea that different strength of norms in societies translate to different game-theoretic interaction structures and incentives. We use this model to study, both analytically and with extensive agent-based simulations, the evolutionary relationships of the \emph{need for coordination} in a society (which is related to its norm strength) with two key aspects of norm change:
 \emph{cultural inertia} (whether or how quickly the population responds when faced with conditions that make a norm change desirable), and \emph{exploration rate} (the willingness of agents to try out new strategies). 
 Our results show that a high need for coordination leads to both high cultural inertia and a low exploration rate, while a low need for coordination leads to low cultural inertia and high exploration rate. This is the first work, to our knowledge, on understanding the evolutionary causal relationships among these factors.
 \end{abstract}

\thispagestyle{firststyle}


\section{Introduction}

{\let\thefootnote\relax\footnotetext{Code for all the simulations presented in this paper is available at: https://github.com/sohamde/inertia-aamas}

Human societies around the world are unique in their ability to develop, maintain, and enforce social norms. Social norms enable individuals in a society to coordinate actions, and are critical in accomplishing different tasks. Neuroscience, field, and experimental research have all established that there are marked differences in the strength of social norms around the globe \cite{balliet2013trust, ensminger2014experimenting, gelfand2011differences, harrington2014tightness, henrich2010markets, henrich2006costly, herrmann2008antisocial, roos2015societal}. Some cultures (e.g., some middle-eastern countries, India, South Korea, etc.) are \emph{tight}, in the sense that they tend to have strong social norms, with a high degree of norm-adherence and higher punishment directed towards norm-violators. Other cultures (e.g., Netherlands, New Zealand, Australia, etc.) are \emph{loose}, i.e., individuals tend to develop weaker norms with more tolerance for deviance \cite{gelfand2011differences, harrington2014tightness, roos2015societal}. This indicates that the nature of human interaction and influence is vastly different across different cultures around the world. 

To date, there has been little research on
the evolutionary processes of norm maintenance and the processes that lead to norm change, and how these processes are substantially different in societies around the world.
However, recent world events (e.g., recent social uprisings and turmoil) show that it is critically important to develop such an understanding \cite{yan2012, latimes2011, reuters2011, nyreview2016}. In this paper, we draw ideas from recent social science research to build culture-sensitive models that provide insights into the substantial societal differences that exist in how individuals interact and influence each other.

Although evolutionary game theory (EGT) was first developed to model biological evolution \cite{hofbauer1998evolutionary, smith1982evolution, weibull1997evolutionary}, it also has become useful as a way to model cultural evolution (for examples, see Section \ref{sec:background}). In this paper, we use EGT to examine the relationships of the amount of \emph{need for coordination}
(which psychological and sociological studies show is related to norm strength \cite{roos2015societal}), with two key aspects of norm change in societies: 
\begin{enumerate*}[label=(\roman*)]
\item the amount of \emph{cultural inertia}, i.e., the amount of resistance to changing a cultural norm, and
\item the \emph{exploration rate}, i.e., the extent to which agents are willing to try out new behaviors.
\end{enumerate*}
More specifically, our primary contributions in this paper are as follows:

\begin{itemize}
\item
We provide a novel way to 
\begin{enumerate*}[label=(\roman*)]
  \item model a society's strength of norms by using an agent's need for coordination in the society, and
  \item model the desirable/undesirable norms in a society.
\end{enumerate*}
This is done by characterizing how they affect the payoffs in a  game-theoretic payoff matrix, leading to different interaction structures and incentives in a society.

\item
We investigate cultural evolution of norm change in this model using two well-known models of change in evolutionary game theory
(the replicator dynamic \cite{taylor1978evolutionary} and the Fermi rule \cite{blume1993statistical}).
Using mathematical analyses and extensive agent-based simulations, we establish that: \emph{the higher the need for coordination is, the higher the cultural inertia will be, and vice versa}. When a population faces conditions that make a norm change desirable, a high need for coordination will make them slower to change to the new norm compared to a society with a lower need for coordination. Further, if the need for coordination is high enough, the existing norm will not change at all.

\item
In order to understand how norms change in different cultures, we also examine whether the need for coordination in a society has a causal evolutionary relationship to an agent's tendency to learn socially (i.e., adopt a behavior that is being used by other agents in the population) versus innovate/explore new random behaviors. In order to be able to do so, we propose a novel way to model this, where we let the exploration rate, i.e., the probability that an agent tries out a new action at random, \emph{evolve} over time as part of the agent's strategy, rather than stay fixed as in previous work \cite{traulsen2009exploration}. 

\item 
The cultural differences in the distribution of agent strategies favoring social learning versus innovation or exploration can have a critical impact on how attitudes, beliefs and behaviors spread throughout the population, and thus, is vital to understanding norm change. At a societal level, such differences can affect the rate at which new technologies, languages, moral traditions, and political institutions are adopted, while at local levels, they can alter the effectiveness of persuasion methods at the individual level. Using the above model of evolving exploration rates, we verify this by establishing, via extensive agent-based simulations, that: \emph{the higher the need for coordination is, the lower the exploration rate will be, and vice versa}.
\end{itemize}

These results provide insight into the reasons why tight societies are less open to change, and why cultural inertia and high levels of social learning develop in such societies. 
To our knowledge, 
this is the first work to provide a culturally-sensitive model of norm change and to show how the processes of norm propagation differ across societies.

The rest of the paper is organized as follows. Section 2 includes background and related work.  Section 3 provides our model of the need for coordination, and mathematical analyses and agent-based simulations showing how it affects cultural inertia. Section 4 describes our model of evolving exploration rates, and shows how the degree of need for coordination affects the evolution of exploration rates.  In Section 5 we discuss the significance of our results.

\section{Background and Related Work}
\label{sec:background}

EGT offers a simple framework for dealing with large populations of interacting individuals, where individuals interact using different strategies, leading to game-theoretic payoffs that denote an individual's evolutionary fitness. EGT was first developed to model biological evolution \cite{hofbauer1998evolutionary, smith1982evolution, weibull1997evolutionary}. In such models, high-fitness individuals are more likely to reproduce than low-fitness individuals, and hence, the strategies used by those high-fitness individuals become more prevalent in the population over time. Thus, EGT studies the evolution of populations, without requiring the usual decision-theoretic `rationality' assumptions typically used in classical game theory models.

EGT models have also been used in studies of a wide variety of social and cultural phenomena \cite{de2016using}, e.g., cooperation \cite{bowles2004evolution, hamilton1981evolution, nowak1992tit, nowak2006five, riolo2001evolution}, punishment \cite{boyd2003evolution, brandt2003punishment, brandt2006punishing, rand2011evolution, roos2014high}, ethnocentrism \cite{de2015inevitability, hammond2006evolution, hartshorn2013evolutionary}, etc. 
In EGT models of cultural evolution, biological reproduction is replaced by social learning: 
if an individual uses some strategy that produces high payoffs, then others are more likely to adopt the same strategy.

EGT models of cultural evolution use highly simplified abstractions of complex human interactions, designed to capture only the essential nature of the interactions of interest. These models do not give exact numeric predictions of what would happen in real life; but they are helpful for studying the underlying dynamics of different social processes, by establishing causal relationships between various factors and observed evolutionary outcomes. Since the evolution of a human culture over time is virtually impossible to study in laboratory settings or field studies, EGT modeling provides a useful tool to apply to the study of culture and norms.


\section{Proposed Model}
\label{proposed_model}

\begin{figure}[tb]
\begin{center}
$M_c = $
\begin{tabular}{ |c|c|c| }
\cline{1-3}
 & $A$ & $B$ \\
\cline{1-3}
$A$ & $a_c,a_c$ & $0,0$ \\
\cline{1-3}
$B$ & $0,0$ & $b_c,b_c$ \\
\cline{1-3}
\end{tabular} 
\hspace{10mm}
$M_f = $
\begin{tabular}{ | c | c | c | }
\cline{1-3}
 & $A$ & $B$ \\
\cline{1-3}
$A$ & $a_f, a_f$ & $ a_f,  b_f$ \\
\cline{1-3}
$B$ & $b_f,  a_f$ & $ b_f, b_f$ \\
\cline{1-3}
\end{tabular} 
\caption{\label{fig:ind_payoffs} Individual payoff matrices. $M_c$ denotes the coordination game and $M_f$ denotes the fixed-payoff game used in our model.}
\end{center}
\end{figure}

Past field and experimental research have shown that tight societies have stronger norms, where individuals adhere to norms much more than loose societies, and face higher punishment when deviating. On the other hand, individuals in loose societies typically have more tolerance for deviant behavior \cite{gelfand2011differences, harrington2014tightness, roos2015societal}. Past EGT studies have shown that a society's exposure to \emph{societal threat} is a key mediating factor in its strength of norms \cite{roos2015societal}, where threats can be either ecological like natural disasters and scarcity of resources, or manmade such as threats of invasions and conflict. In high-threat situations, societies tend to develop strong norms for coordinating social interaction, (i.e., to become tighter), since coordination is vital for the society's survival. In low-threat situations, there is less need for coordination, which affords weaker norms and looser societies.

\begin{figure*}[tb]
\begin{center}
\begin{tabular}{ c }
$M$ =
\begin{tabular}{ |c|c|c| }
\cline{1-3}
 & $A$ & $B$ \\
\cline{1-3}
$A$ & $ca_c + (1-c)a_f, ca_c + (1-c)a_f$ & $(1-c)a_f,(1-c)b_f$ \\
\cline{1-3}
$B$ & $(1-c)b_f,(1-c)a_f$ & $cb_c + (1-c)b_f, cb_c + (1-c)b_f$ \\
\cline{1-3}
\end{tabular}
\end{tabular} 
\caption{\label{tab:weighted_payoff} Weighted payoff matrix $M$ used in our model defined as $M = cM_c + (1-c)M_f$.}
\end{center}
\end{figure*}

Using this intuition, we hypothesize that individuals in different societies interact using \emph{different} payoff structures and incentives. Tight societies tend to have a high need for coordination, and we can model the extreme case as a \emph{coordination game} $M_c$, where one only gets a payoff if playing the same action as the agent one is interacting with. In loose societies, on the other hand, individuals' payoffs are less affected by others' actions, and we can model the extreme case as a \emph{fixed-payoff} game $M_f$, in which an agent's payoff depends only on the action played by that agent, and not on the actions of the other agent. For cases in between the two extremes, we use a game in which the payoff matrix is a weighted combination of a coordination game and a fixed-payoff game, with the weighting factor $0 \leq c \leq 1$ denoting the need for coordination.

As is done in many EGT studies, we consider games in which individuals have two possible actions to choose from. In our case, the two actions $A$ and $B$ correspond to possible norms that the society could settle on. 
As shown in Figure \ref{fig:ind_payoffs}, the coordination game
has a payoff matrix $M_c$ in which $a_c$ and $b_c$ are the payoff parameters; and the fixed-payoff game
has a payoff matrix $M_f$ in which $a_f$ and $b_f$
are the payoff parameters.
The weighted combination of the two games, shown in Figure \ref{tab:weighted_payoff}, is
$M = cM_c + (1-c)M_f$, where $0 \leq c \leq 1$ is the need for coordination.
\begin{figure}
\begin{center}
$M' = $
\begin{tabular}{ |c|c|c| }
\cline{1-3}
 & $A$ & $B$ \\
\cline{1-3}
$A$ & $a$, $a$ & $(1-c)a$, $(1-c)b$ \\
\cline{1-3}
$B$ & $(1-c)b$, $(1-c)a$ & $b$, $b$ \\
\cline{1-3}
\end{tabular} 
\caption{Updated payoff matrix after assuming $a_c - b_c = a_f - b_c$ and adding a suitable constant to the payoffs in $M$ in Figure \ref{tab:weighted_payoff}.}
\label{fig:new_payoffs}
\end{center}
\end{figure}

We first present a lemma that shows that under a mild assumption, the payoff matrix $M$ can be much simplified on adding a constant to all payoffs in the matrix. \\

\begin{lemma}
\label{affine_lemma}
Consider the game matrix $M$ defined in Figure \ref{tab:weighted_payoff}, and assume that $a_c - b_c = a_f - b_f$. Then, under a suitable addition of a constant to the payoffs, and denoting $a_c = a$ and $b_c = b$, the game matrix $M$ reduces to the matrix $M'$ shown in Figure \ref{fig:new_payoffs}.
\end{lemma}
\begin{proof}
Let the constant $\gamma$ be defined as $\gamma = a_c - a_f = b_c - b_f$ (where the second equality holds under the assumption). On adding $\gamma$ to all payoffs in $M$, the payoff matrix $M$ reduces to $M'$, shown in Figure \ref{fig:new_payoffs}, where we denote $a_c = a$ and $b_c = b$.
\end{proof}

The assumption $a_c - b_c = a_f - b_f$ is very reasonable, since this just ensures that switching from one norm to the other always results in the same change in payoffs, regardless of the weight $c$ on the coordination game. Otherwise, there would be an added causal factor for the dynamics of norm change. Also note that, from Lemma \ref{affine_lemma}, under additions with a constant, this assumption reduces to just setting $a_c = a_f$ and $b_c = b_f$. For the rest of the paper, we will work with payoff matrix $M'$ where we set $a_c = a_f = a$ and $b_c = b_f = b$. In subsequent sections, we will show why simplifying the payoff matrix by adding a constant value to all payoffs (as shown in Lemma \ref{affine_lemma}) is a perfectly reasonable step to take.


From payoff matrix $M'$, we see that whenever $b < a$, the better action for the society to settle on (in terms of payoff) is $A$, while if $a < b$ then it is $B$. Let $M'_{AB}$ be the payoff that an agent receives when they play action $A$ and their opponent plays action $B$. Let $M'_{AA}$, $M'_{BA}$ and $M'_{BB}$ be defined similarly.  Studying the Nash equilibrium of the game $M'$, we get the following lemma. The proof is presented in the appendix. \\


\begin{lemma}
\label{two_player_nash}
Consider the game matrix $M'$ defined in Figure \ref{fig:new_payoffs}, where all payoff values are positive, i.e., $a, b > 0$. Then we have:
\begin{enumerate}[label=(\roman*)]
\item If $b > a$, the strategy profile $(B, B)$ is a Nash Equilibrium. Further, if $c \ge \frac{b - a}{b}$, then $(A, A)$ is also a Nash equilibrium. Further, the strategy profile $((q, 1-q), (q, 1-q))$ is a Nash Equilibrium only when $c \ge \frac{b - a}{b}$, where $q = \frac{b - (1-c)a}{c(a + b)}$. Note that the mixed strategy $(q, 1-q)$ denotes playing action $A$ with probability $q$ and action $B$ with probability $1-q$.
\item Similarly, if $a > b$, the strategy profile $(A, A)$ is a Nash Equilibrium. Further, if $c \ge \frac{a - b}{a}$, then the strategy profile $(B, B)$, as well as $((q, 1-q), (q, 1-q))$ are also Nash Equilibria, with $q = \frac{b - (1-c)a}{c(a + b)}$. \\
\end{enumerate}
\end{lemma}

From Lemma \ref{two_player_nash}, we see that only when $c$ is high enough, the \emph{sub-optimal} action pair becomes a Nash Equilibrium, where \emph{sub-optimal} refers to the fact that the action pair that has lower payoff than the optimal action pair. This means that when $b > a$, $(A, A)$ is the sub-optimal action pair. Thus, from Lemma \ref{two_player_nash}, we that see if the need for coordination $c$ is high, then the population may converge to either of two different equilibria, one of which is sub-optimal in terms of overall payoff.
When $c$ is low, on the other hand, the society will converge to a single globally-optimal equilibrium.

In the next two sub-sections we introduce two models for studying norm change, using two well-known models of evolutionary change (the replicator dynamic \cite{taylor1978evolutionary} and the Fermi rule \cite{blume1993statistical}). We show that both models of evolutionary change are invariant to additions to the payoffs by a constant, and thus the results from this section carry forward. We present results for how different societies respond to a need for norm change using both mathematical analysis on infinite \emph{well-mixed} populations (where well-mixed denotes that any agent can interact with any other agent in the population), and extensive agent-based simulations on finite structured populations (where agents are placed on a network and can interact with only their neighbors).

\subsection{Replicator dynamic on infinite well-mixed populations}
\label{sec:well_mixed}

Consider a well-mixed infinite population of agents. This is a standard setting used in evolutionary game theory, since a well-mixed infinite population is usually analytically tractable. Let the agents be interacting with each other using game matrix $M'$ defined in Figure \ref{fig:new_payoffs}, and the proportion of agents playing each strategy be denoted by $x = (x_A, x_B)$, i.e., $x_A$ proportion of agents with strategy $A$, and proportion $x_B = 1 - x_A$ with strategy $B$. Also, let $u_A(x)$ and $u_B(x)$ denote the payoffs received by an agent playing actions $A$ and $B$ respectively. The expected payoff for an agent is given by interacting with a randomly chosen agent in the population. Thus, we get the following:
\begin{align*}
\mathbb{E} [u_A(x)] &= x_A M'_{AA} + x_B M'_{AB},  \\
\mathbb{E} [u_B(x)] &= x_A M'_{BA} + x_B M'_{BB}.
\end{align*}
On analyzing the Nash Equilibria of this system, we observe the following lemma. The proof is presented in the appendix. \\

\begin{lemma}
\label{well_mixed_nash}
Consider a well-mixed infinite population where agents interact using the game $M'$ in Figure \ref{fig:new_payoffs}. Assuming all payoff values are positive, i.e., $a, b > 0$, and using Lemma \ref{two_player_nash}, we have:
\begin{enumerate}[label=(\roman*)]
\item When $b > a$, $x_A = 0$ is a Nash Equilibrium. If $c \ge \frac{b - a}{b}$, then $x_A = 1$ and $x_A =  \frac{b - (1-c)a}{c(a + b)}$ (which corresponds to the mixed-strategy Nash Equilibrium in Lemma \ref{two_player_nash}) are also Nash Equilibria.
\item Similarly, when $a > b$, $x_A = 1$ is a Nash Equilibrium, while if $c \ge \frac{a - b}{a}$, then $x_A = 0$ and $x_A = \frac{b - (1-c)a}{c(a + b)}$ also are Nash Equilibria. \\
\end{enumerate}
\end{lemma}

We assume that on each iteration, agents interact with other randomly chosen agents, and the population evolves according to the replicator dynamic. 
 The replicator dynamic is based on the idea that the proportion of agents of a type (or strategy) increases when it  achieves expected payoff higher than the average payoff, and decreases when achieving lower payoff than the average payoff. Thus, over time, the proportion of agents of a type that achieves payoff higher than the average payoff starts increasing in the population, and eventually take over. More formally, the replicator dynamic is given by the differential equation
\begin{align}
\dot{x}_A =\frac{dx_A}{dt} = x_A \cdot (\mathbb{E}[u_A(x)] - \theta(x)), \label{replicator_inf}
\end{align}
where $\theta(x) = x_A \mathbb{E}[u_A(x)] + x_B \mathbb{E}[u_B(x)]$ is the average payoff received by all agents in the population. From \eqref{replicator_inf}, it is clear that the rate of change remains the same on adding a constant to the payoff matrix, since the added constants would just cancel each other out. Thus, Lemma \ref{affine_lemma} follows through to this section as well.

Using the game matrix $M'$, the rate of change in the proportion $x_A$ is given by:
\begin{align}
\dot{x}_A  = x_A (1-x_A) (c(a + b)x_A - (b - (1-c)a)).
\label{replicator}
\end{align}
The fixed points of this rate of change are given by:
\begin{align}
\label{fixed_points}
x_A = 0, \quad x_A =1, \quad \text{and} \quad x_A = \frac{b - (1-c)a}{c(a + b)}.
\end{align}
These correspond to the Nash Equilibria derived earlier. Next, we study the stability of the Nash equilibria derived above, where we define a stable Nash equilibrium under the replicator dynamic to be one where: if an infinitesimal proportion of agents change their strategy, the replicator dynamic always forces the population back to the original Nash equilibrium. More precisely, let the Nash equilibrium be $x_A = p$. If $x_A$ increases an infinitesimal amount to $p+\epsilon$, the Nash equilibrium is stable only if $\dot{x}_A < 0$, which drives the population back to the Nash equilibrium $x_A = p$. Similarly, if $x_A$ decreases by $\epsilon$, $x_A = p$ is stable only if $\dot{x}_A > 0$. 
Thus, we state the following corollary. \\



\begin{corollary}
\label{stability}
From Lemma \ref{well_mixed_nash} and  Eq.\ \eqref{replicator} and  Eq.\ \eqref{fixed_points}, we see that the Nash Equilibria $x_A = 0$ and $x_A = 1$ are stable, while the Nash Equilibrium $x_A = \frac{b - (1-c)a}{c(a + b)}$ is unstable.
\end{corollary}
\begin{proof}
Let $\phi =  \frac{b - (1-c)a}{c(a + b)}$. From Eq.\ \eqref{replicator}, we notice that, if $x_A = \phi + \epsilon$, then $\dot{x}_A > 0$, while if $x_A = \phi - \epsilon$, then $\dot{x}_A < 0$, for any small $\epsilon > 0$. Thus, $x_A = \frac{b - (1-c)a}{c(a + b)}$ represents an unstable fixed point. Similarly notice that if $x_A = \epsilon$, $\dot{x}_A < 0$, while if $x_A = 1 - \epsilon$, $\dot{x}_A > 0$. Thus, $x_A = 0$ and $x_A =1$ represent stable fixed points.
\end{proof}

There is a further notion of equilibrium used in EGT called evolutionarily stable strategies (ESS) \cite{smith1982evolution}. A strategy $S$ is an ESS if there is a small proportion $p_y$ such that, when any other strategy $T$ has a proportion $p_x < p_y$ (where the rest of the population has strategy $S$), the payoff of an $S$ agent is always strictly greater than a $T$ agent. Using this definition, we state the following theorem. The proof is presented in the appendix. \\



\begin{theorem}
\label{ess_thm}
From Lemma \ref{well_mixed_nash} and Corollary \ref{stability}, we see:
\begin{enumerate}[label=(\roman*)]
\item When $b > a$, $B$ is an ESS. If $c \ge \frac{b - a}{b}$, then $A$ is also an ESS.
\item When $a > b$, $A$ is an ESS. If $c \ge \frac{a - b}{a}$, then $B$ is also an ESS. \\
\end{enumerate}
\end{theorem}

\begin{figure*}[tb]
        \centering
        \begin{subfigure}[h]{0.32\textwidth}
                \includegraphics[width=\textwidth]{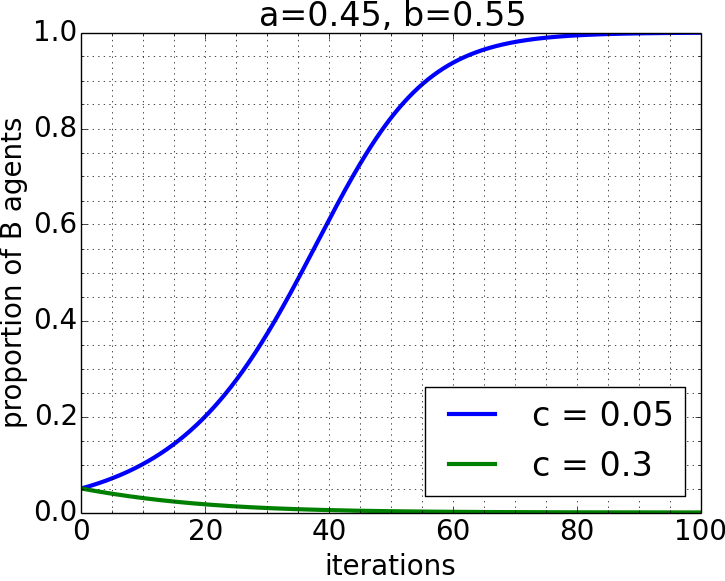}
        \end{subfigure}
        \begin{subfigure}[h]{0.32\textwidth}
                \includegraphics[width=\textwidth]{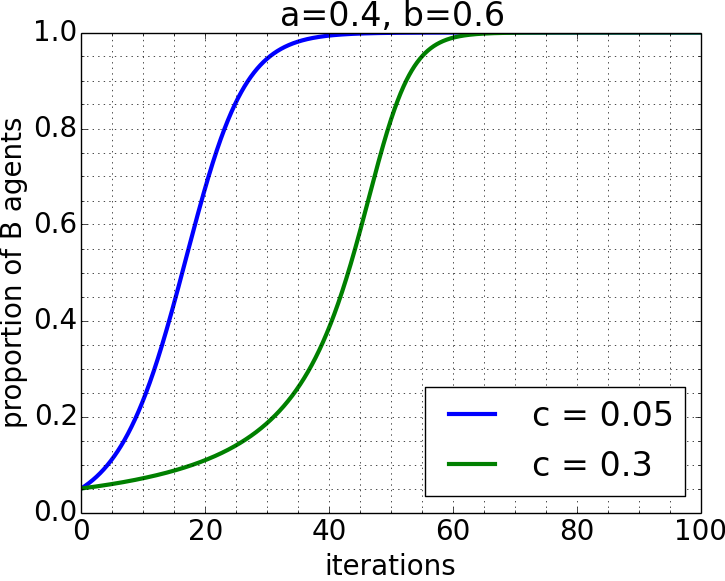}
        \end{subfigure}
        \begin{subfigure}[h]{0.32\textwidth}
                \includegraphics[width=\textwidth]{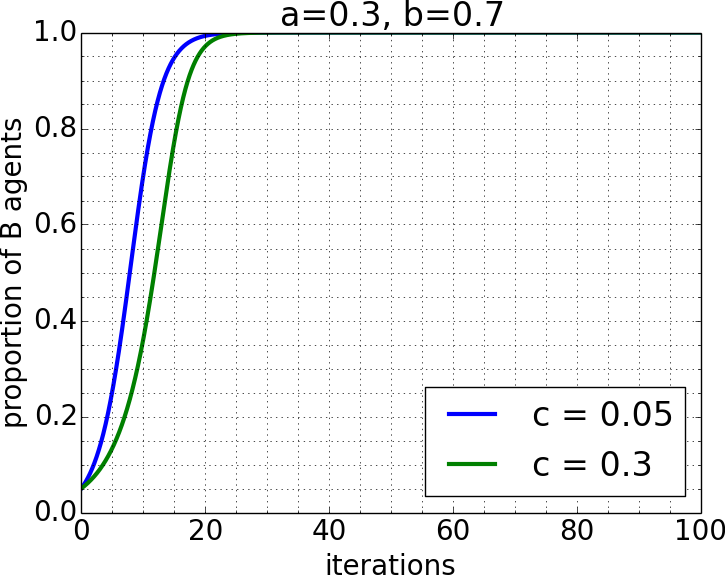}
        \end{subfigure}
         \caption{Figures show the change in the proportion of $B$ agents with time with a well-mixed infinite population where reproduction is determined by the replicator dynamic with $b > a$.}
         \label{well_mixed_inertia_sims}
\end{figure*}


%
We observe  that the strategies $A$ and $B$ are Evolutionary Stable Strategies (ESS), when adopted by everyone in the population (corresponding to the stable Nash equilibria $x_A = 1$ and $x_A = 0$). The unstable Nash Equilibrium, on the other hand, does not correspond to an ESS, since even a small group with a different strategy is able to force the population to a different equilibrium. Thus, only stable Nash Equilibria correspond to evolutionarily stable strategies.

Theorem \ref{ess_thm} indicates that a society is bound to end up at one of the evolutionarily stable strategies: with every individual on action $A$ or everyone on action $B$, since even a small perturbation moves the society away from the unstable Nash equilibrium. When $c$ is low, there exists only a single ESS, and thus the society adapts itself and settles on the ESS. When $c$ is high, there are two ESSs, and thus the society might settle on either one, depending on the starting point of the society.

\begin{figure}[tb]
        \centering
\includegraphics[width=0.4\textwidth]{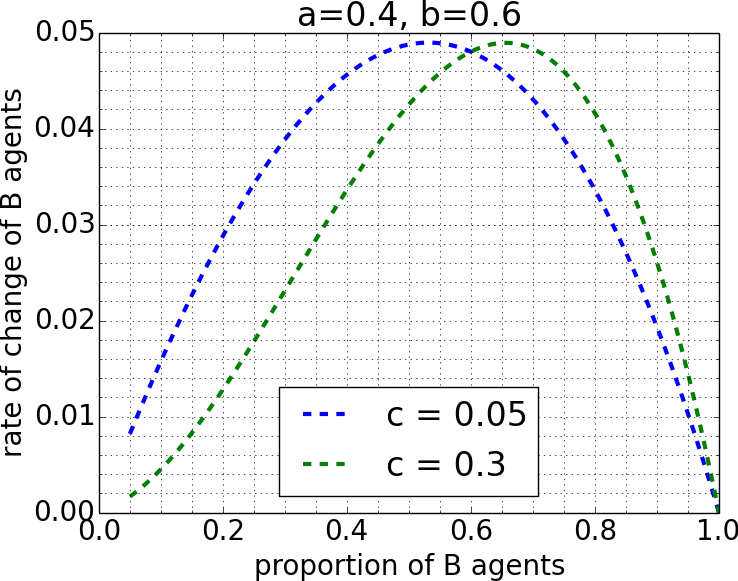}
         \caption{Figure shows the rate of change of $B$ agents versus the proportion of $B$ agents, with a well-mixed infinite population where reproduction is determined by the replicator dynamic with $b > a$.}
         \label{well_mixed_rate_sims}
\end{figure}

Let us consider two societies: one with a lower need for coordination $c_1$, and one with a high need for coordination $c_2 > c_1$. 
To avoid some awkward phrasing, we'll call these the ``looser'' and ``tighter'' societies, respectively.
Suppose a majority of both societies are playing norm $A$, and suppose they evolve according to the replicator dynamic given in Eq.\ \eqref{replicator}. We are interested in how these two societies would respond to
 the action $B$,
 when the payoff of action $B$ is higher than $A$, i.e., when $b > a$, or equivalently,
$M'_{BB} > M'_{AA}$. First notice that if $c_2 > (b - a)/b$, and $c_1 < (b - a)/b$,
it follows from Theorem \ref{ess_thm} that the tighter society remains on norm $A$ while the looser one switches to the globally optimal norm $B$.

Now suppose the difference in norm payoffs is large enough such that $c_2 < (b - a)/b$ (and thus also,  $c_1 < (b - a)/b$). This ensures that there is only a single equilibrium for both societies at $x_A = 0$. Thus, both societies would switch to norm $B$, and we are interested in the rate at which this change occurs. Let $\dot{x}_{B1}$ and $\dot{x}_{B2}$ denote the rate of change when the need to conform is $c_1$ or $c_2$, respectively. Then we can show that:
\begin{align*}
\dot{x}_{B2} - \dot{x}_{B1} = x_B (1 - x_B) (c_2 - c_1)((a + b)x_B - b).
\end{align*}
This simplifies to
\begin{align}
\dot{x}_{B2} - \dot{x}_{B1} \begin{cases}
	\le 0,& \text{when } x_B \le \frac{b}{a + b};\\
	> 0,& \text{when } x_B > \frac{b}{a + b}.\\
	\end{cases}
\label{rate_difference}
\end{align}
Thus, $\dot{x}_{B2} < \dot{x}_{B1}$ in the initial stages when $x_B < b/(a + b)$. However, once the proportion of $B$ agents become big enough such that $x_B > b/(a + b)$, then the higher the value of $c$, the higher the rate of change will be. Thus, when $c$ is high, the switch from $A$ to $B$ takes time to speed up, with more cultural inertia than when $c$ is low, even when the payoff of the new norm is arbitrarily large compared to the previous norm. The initial cultural inertia results in the society with a higher $c$ value to take longer overall to switch to the new norm.

Figure \ref{well_mixed_inertia_sims} illustrates these properties of well-mixed populations using the replicator dynamic. We start off the society at the proportion $x_A = 0.95$. In the first of the three graphs, the tighter society (again using ``tighter'' as shorthand for ``higher need for coordination'') 
has $c > (b - a)/b$. Thus, while the less-tight society switches to the more beneficial norm $B$, the tighter society is resistant to the change (since the difference in payoffs is small) and stays with norm $A$. The second and third graphs show situations where both societies switch to norm $B$. We observe that the tighter society switches more slowly towards changing to norm $B$, but the difference in speed decreases as the difference in payoffs between $B$ and $A$ increases. 

As derived in Eq.\ \eqref{rate_difference}, 
the rate of change for a society with higher $c$ grows larger than with lower $c$ only after $x_B > \frac{b}{a + b}$. This is shown in Figure \ref{well_mixed_rate_sims}. This also indicates the initial inertia that societies with a higher need for coordination experience towards changing norms. The need for coordination in these societies lead to individuals being reluctant to try out new norms, which in turn leads to inertia. 


\begin{figure}[t]
\centering
\includegraphics[width=0.6\textwidth]{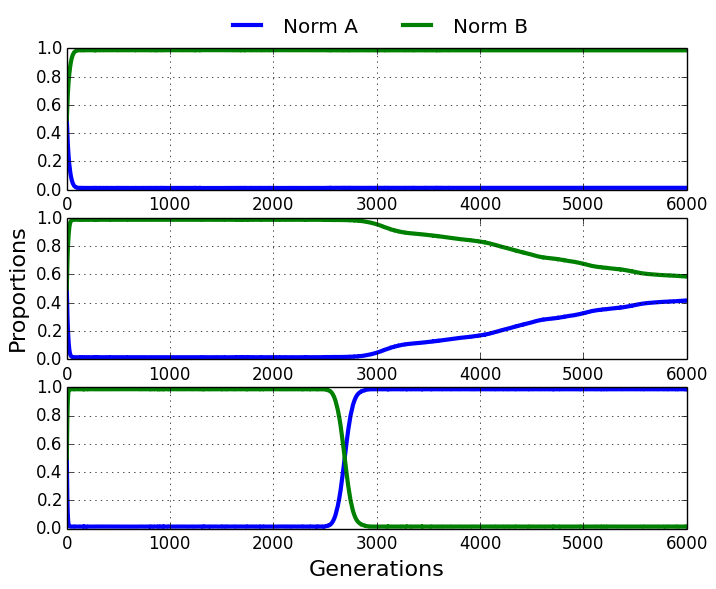}
\caption{Simulations with the Fermi rule on a toroidal grid of size 2500. From top to bottom: $c = 1.0$, $c = 0.75$, $c = 0.5$. Initially: $a = 1.0$, $b = 1.15$. We use a structural shock at 2500 iterations, after which the payoffs become: $a =1.15$, $b = 1.0$.}
\label{game_model_inertia}
\end{figure}

\subsection{Agent simulations on finite networks}
\label{finite_network_inertia}

A limitation of the above model is that it assumes that the population is \emph{infinite} and \emph{well-mixed}. While the assumption that a population is infinite is not a bad approximation for very large populations (which is the scale that we are interested in), the assumption that agents are well-mixed, i.e., where any agent can interact with any other agent, is often inaccurate. In this section, we show that the results derived in the previous section, also extend to cases where agents are structured on the nodes of a graph/network, where agents can only interact with another agent if they are connected by an edge in the graph.

More specifically, we now consider a structured population where agents are arranged on the nodes of a toroidal (wrap-around) grid, such that each agent can interact only with the 4 other agents they are connected to. 
We consider toroidal grids as a convenient example, however, the results we describe below also extend to other network structures like small-world networks \cite{watts1998collective}, and preferential attachment \cite{barabasi1999emergence} models.
Mathematical analysis of evolutionary games on structured populations is not yet a well-developed field, and thus we perform simulations of our model as follows.

Initially, we arrange agents with random strategies ($A$ or $B$) on each node of the grid. In each iteration, each pair of agents connected by an edge interact in a two-player game defined by the payoff matrix $M$. The total payoff of each agent is computed by summing over the payoffs received by an agent for each game that they play. Since the population is finite, we use dynamics defined on finite populations. After each interaction phase, agents use the Fermi rule to update it's strategy for the next iteration. Under the Fermi rule, an agent $\psi_a$ picks a random neighbor $\psi_n$ and observes its payoff, and the agent then decides to switch to the neighbor's strategy with probability
$p = (1+\exp(-s(u_a - u_n)))^{-1}$,
where $u_a$ and $u_n$ are the payoffs of the agent and the neighbor, and $s$ is a user-defined parameter (in all our experiments, we set $s= 5$). With probability $1-p$, the agent retains its old strategy. With a small probability $\mu$, called the exploration rate, an agent then tries out an action completely randomly. This repeats for every iteration of the simulation. Note that the Fermi rule also only depends on a difference between payoff values, and thus, like the replicator dynamic, is also invariant to addition of a constant to the payoffs. Thus, for all our experiments in this section, we use the simplified game matrix $M'$ from Figure \ref{fig:new_payoffs}.

To study cultural inertia (i.e., resistance to changing a cultural norm) or rapid cultural change in different societies, we use a game-theoretic model of a \emph{structural shock}. A structural shock represents a catastrophic incident in a society, where suddenly there is abrupt change in the payoffs for actions $A$ and $B$. We are interested in studying how societies with different needs for coordination react to such an abrupt and drastic shift in the payoffs of the possible actions.
In our EGT model, we implement a structural shock by simply interchanging the payoffs of actions $A$ and $B$, thus, denoting a sudden change in the globally optimal action in a society. This is equivalent to interchanging the payoff values $a$ and $b$. Thus, if initially, we have $b > a$, after a structural shock, we get $a > b$.

Now consider that, initially, the action with a higher utility (and the current norm) in a society is $B$, i.e., $b > a$ with $x_A = 0$. Suppose, the society experiences a structural shock, where now action $A$ becomes more desirable with $a > b$. On introducing a small proportion of agents playing norm $A$ (say $x_A = 0.01$), 
 if the need for coordination is low then the population will switch to the new norm with $x_A = 1$. This is because, after the structural shock, the Nash Equilibrium (and ESS) is $x_A = 1$, as shown above. However, if the need for coordination is high (i.e., $c \ge \frac{a - b}{a}$), then $x_A = 0$ is still a Nash Equilibrium (and ESS) and the population will remain on the sub-optimal norm $B$ even after the structural shock. 

All experiments were run on a grid with 2500 nodes, and the simulation goes on for 6000 iterations, with a structural shock implemented at 2500 iterations. 100 independent simulations are run for each setting and the results are averaged over the 100 runs. Figure \ref{game_model_inertia} shows the results of our simulations. The plots show the proportion of agents playing norm $A$ vs norm $B$. As before, the parameter $c$ denotes the need for coordination.
When $c$ is low, very little cultural inertia develops and agents are more willing to innovate by exploring behaviors other than the current societal norms. In this case, the population will change more quickly to a different norm if the new norm will be beneficial. By contrast, when $c$ is high, we see the evolutionary emergence of higher levels of cultural inertia, with agents less willing to innovate or to violate established cultural norms. In this case, the population is slower to change to the new norm, and if $c$ is high enough it may not change at all. Thus, qualitatively, the results with a structured populations match those from the infinite well-mixed populations in Section \ref{sec:well_mixed}, and the mechanics that lead to the above results can be explained using the same equilibrium results derived above.

\section{Evolving Exploration Rates}

In addition to cultural inertia, another key aspect to understanding how norms change is to study an agent's tendency to learn socially (i.e., adopt behaviors used by others in the population) or innovate and explore new behaviors. Such tendencies are critical in understanding the rate at which new technologies, languages, or moral traditions are adopted in a population, and provide insight about the processes of influence and persuasion at the individual level.

In the model presented in Section \ref{finite_network_inertia} for finite structured populations, the exploration rate (i.e., the small probability with which an agent tries out a new strategy at random) was kept at a constant low value. This exploration rate denotes how much an agent is open to change and trying out new actions at random. Thus, it seems that the need for coordination in a society might affect how likely an individual is to try out different actions, instead of conforming to their neighbors. Particularly, it seems natural to assume that individuals in tight societies are much less likely to try out random actions than individuals in loose societies \cite{gelfand2011differences, harrington2014tightness}. In this section, we test this hypothesis by presenting a model to study the evolution of exploration rates in different societies.

\begin{figure}[t]
        \centering
        \begin{subfigure}[h]{0.4\textwidth}
                \includegraphics[width=\textwidth]{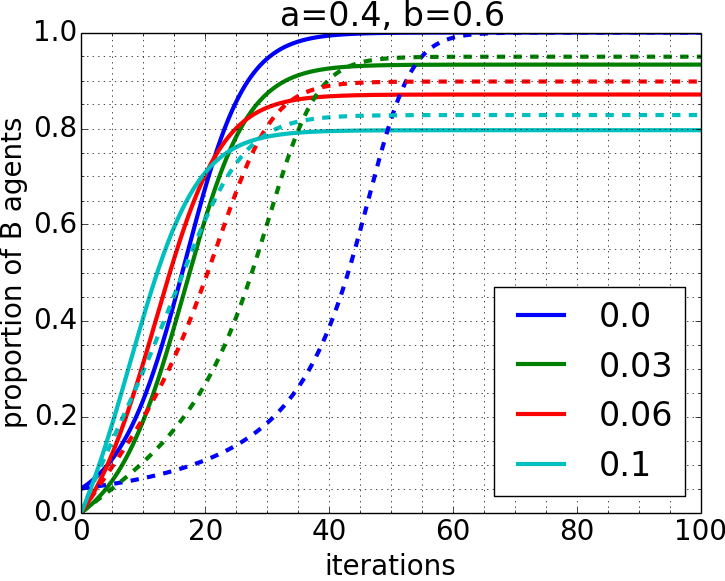}
        \end{subfigure}
        \begin{subfigure}[h]{0.4\textwidth}
                \includegraphics[width=\textwidth]{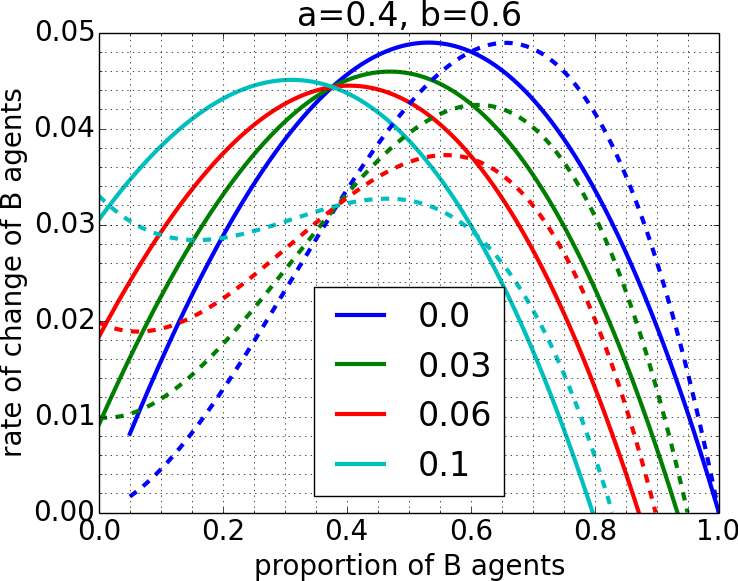}
        \end{subfigure}
         \caption{Replicator-mutator dynamic on an infinite well-mixed population with $a = 0.4$ and $b = 0.6$. The solid and dotted lines denote $c = 0.05$ and $c = 0.3$, respectively. The colors denote the exploration rates.}
         \label{exploration_well_mixed}
\end{figure}

To get some intuition about the hypothesis,
we go back to our setting of a well-mixed infinite population. Note that the replicator dynamic does not have a provision for exploration rates. Thus, we use a variant of the replicator dynamic called the replicator-mutator equation \cite{levin2003complex}. Using this variant, one can include exploration rates into the replicator dynamic. Thus, if we fix $\mu$ to be the exploration rate, we can write the replicator-mutator equation as:
\begin{align*}
\dot{x}_A &= (1- \mu) x_A \mathbb{E}[u_A(x)] + \mu x_A \mathbb{E}[u_B(x)] - x_A \theta(x) \\
&= x_A( \mathbb{E}[u_A(x)] - \theta(x)) + \mu x_A (\mathbb{E}[u_B(x)] - \mathbb{E}[u_A(x)]).
\end{align*}
Thus, like the replicator dynamic, we can write the rate of change in terms of payoff differences, which makes the dynamic invariant to additions to the payoffs by a constant. Thus, we again use the simplified game matrix $M'$ from Figure \ref{fig:new_payoffs}.
Simplifying the equation for $\dot{x}_A$, we get:
\begin{align}
\dot{x}_A  =& \, \, x_A (1-x_A) (c(a + b)x_A - (b - (1-c)a)) \nonumber \\
&+ \mu (x_A x_B (1-c) (b - a) + (x_B^2 b - x_A^2 a)). \label{mut_repl}
\end{align}

Figure \ref{exploration_well_mixed} plots the replicator-mutator equation (Eq.\ \eqref{mut_repl}) with a well-mixed infinite population. The solid lines are for a low need for coordination ($c = 0.05$), while the dotted lines are for a high need for coordination ($c = 0.3$), and we plot the proportion of $B$ agents, as well as the rate of change, for various exploration rate $\mu$ values. From the figure, we see that for all exploration rates $\mu$, when the need for coordination is high then there is higher cultural inertia.

\begin{figure*}[t]
        \centering
        \begin{subfigure}[h]{0.32\textwidth}
                \includegraphics[width=\textwidth]{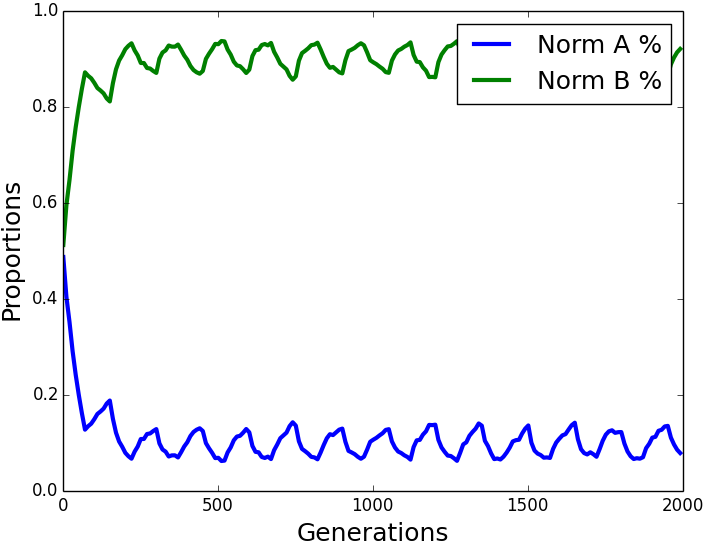}
        \end{subfigure}
        \begin{subfigure}[h]{0.32\textwidth}
                \includegraphics[width=\textwidth]{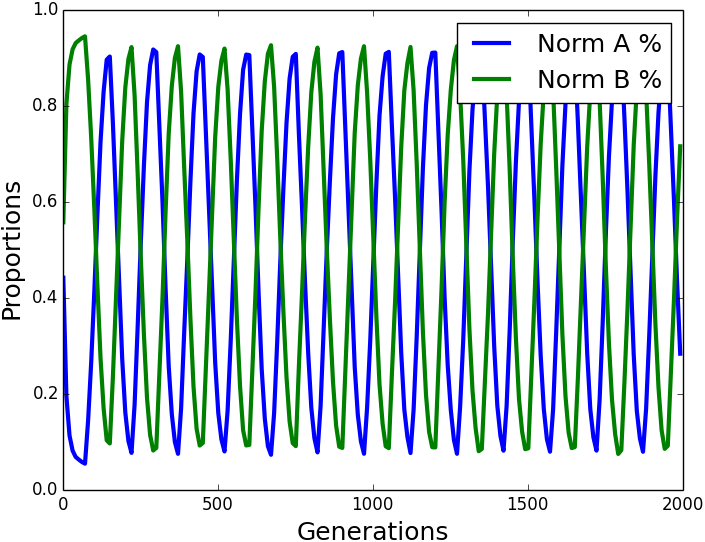}
        \end{subfigure}
        \begin{subfigure}[h]{0.32\textwidth}
                \includegraphics[width=\textwidth]{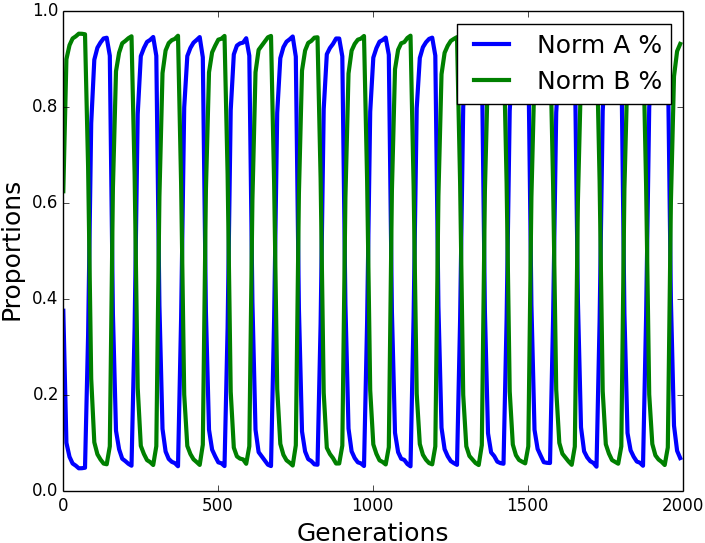}
        \end{subfigure}
        \begin{subfigure}[h]{0.32\textwidth}
                \includegraphics[width=\textwidth]{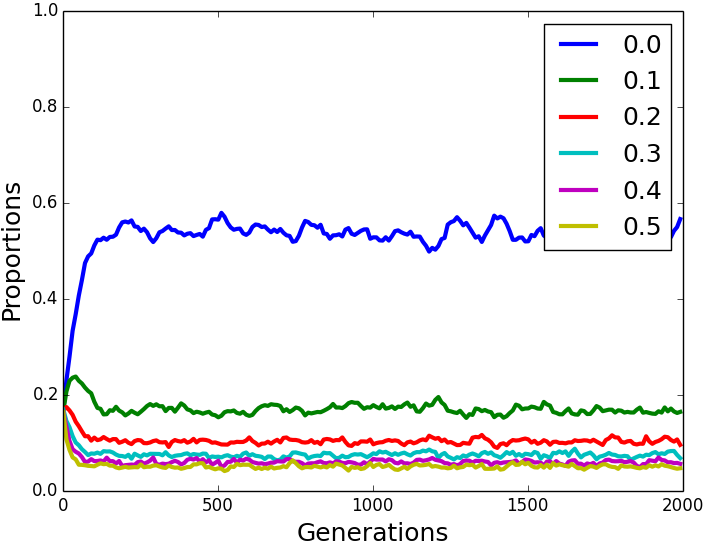}
                \caption{$c = 1.0$}
                \label{fig:sl_c1_mu}
        \end{subfigure}
        \begin{subfigure}[h]{0.32\textwidth}
                \includegraphics[width=\textwidth]{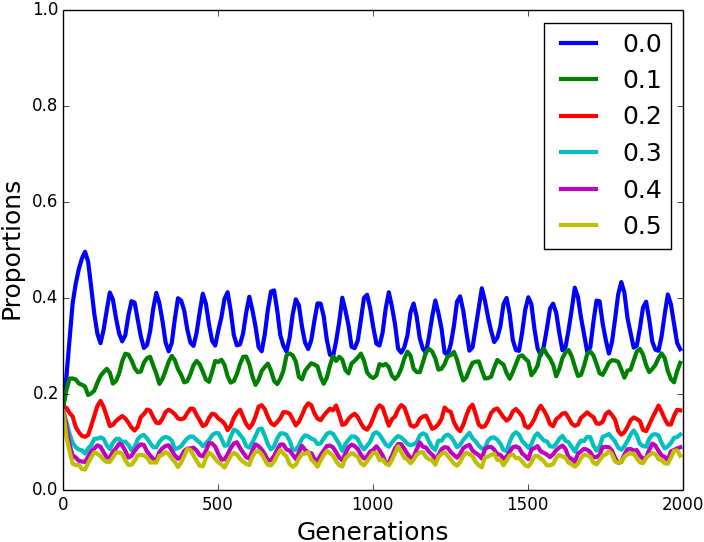}
                \caption{$c = 0.8$}
                \label{fig:sl_c08_mu}
        \end{subfigure}
        \begin{subfigure}[h]{0.32\textwidth}
                \includegraphics[width=\textwidth]{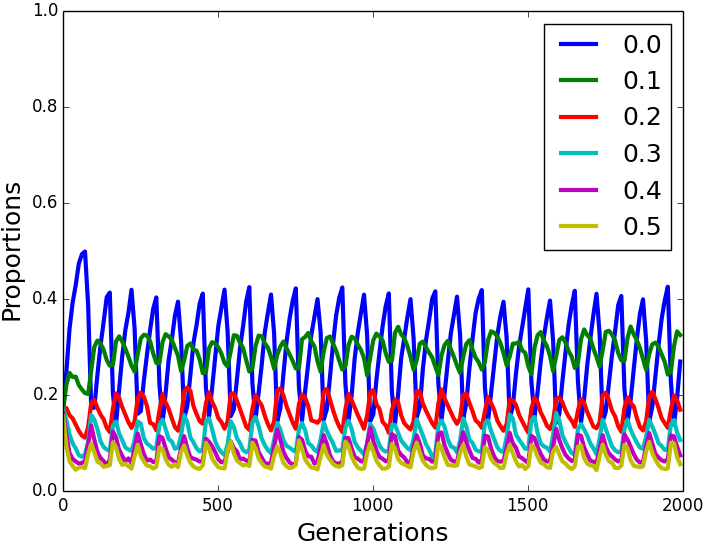}
                \caption{$c = 0.5$}
                \label{fig:sl_c05_mu}
        \end{subfigure}
         \caption{Simulations with the Fermi rule on a toroidal grid of size 2500, with structural shocks at intervals of 75 iterations. From left to right: $c = 1.0$, $c = 0.8$, $c = 0.5$. Initially: $a = 1.0$, $b = 1.15$. The top row shows proportions of norms $A$ and $B$. The bottom row shows proportions of the population using different exploration rates.}
         \label{exploration_sims}
\end{figure*}

To study the evolution of explorative behaviors, we let the exploration rate (referred to as the mutation rate
in biological models) \emph{evolve}. The exploration rate is the probability $\mu$ with which an agent chooses a random new strategy at each iteration ($0 \le \mu \ll 1$). In biological evolution, mutation occurs so rarely that game-theoretic biological models often omit it. In cultural evolution, however, exploration is an important step since individuals try out new behaviors much more frequently \cite{traulsen2009exploration}. Studying the evolution of exploration rates helps us get insights about a society's openness to change. Low exploration rates suggest that individuals are less likely to try out new strategies and are more likely to coordinate with their neighbors. On the other hand, high exploration rates mean that individuals are more open to change and innovation.

To model the evolution of exploration rates, we first create a set $L$ of possible exploration rates. These can be a finite discrete set of exploration rates. For all our experiments, we use the set of exploration rates: $L = \{0.0, 0.1, 0.2, 0.3, 0.4, 0.5\}$. The exploration rate is added as part of the strategy of an agent, and each individual now chooses an exploration rate in addition to the game action ($A$ or $B$). Thus, an agent now copies the exploration rate of a neighbor, along with the game action, when updating its strategy using the Fermi rule. 

Note that, a regularly changing environment is essential for studying the evolution of exploration rates since, if the environment is not changing frequently enough, an exploration rate of 0 would always be evolutionarily stable. To model the changing environment, we will use the same switch in dominant norms (structural shock) that we used in our earlier experiments, except now we apply the structural shock multiple times at much shorter and regular intervals. We use a fixed interval of 75 iterations to apply the structural shock. We run the simulation for a total of 2000 iterations. For these experiments, an agent's strategy set now becomes a size of 10: 5 possible exploration rates in $L$ $\times$ 2 possible game actions (norm $A$ or norm $B$). We use the same toroidal grid as described before. Figure \ref{exploration_sims} shows the experimental results. Each column in Figure \ref{exploration_sims} shows, for a specific $c$, the proportion of agents playing norm $A$ vs norm $B$ (top plot), and the proportion of agents with each exploration rate (bottom plot).

We see that when the need for coordination is high, low exploration rates are adopted by the majority of the society. Individuals in such a society are more likely to adopt the strategies of their neighbors, and this leads to high cultural inertia. In loose societies, however, higher proportions of exploration rates $\mu >0$ evolve, and individuals are more open towards change, leading to lower cultural inertia. This fits well with our results in Section \ref{proposed_model}, and provides insights into why cultural inertia develops in societies with a higher need for coordination.

\section{Discussion}



In this paper, we examined the processes underlying cultural inertia and norm change. We build evolutionary game-theoretic models that show that societies that have a higher need for coordination -- those that are tight -- have higher cultural inertia, with individuals being less likely to switch to the new norm even when it might have a larger payoff. 
Societies with a lower need for coordination -- those that are loose -- on the other hand, have low cultural inertia, with individuals more willing to innovate and open to change. 

By letting the exploration rate evolve, we use it to study an agent's tendency to either learn using social interaction or innovate and explore new random behaviors, and we show that exploration rates evolve differently in different cultures. 
When the need for coordination is high, the majority of the population has very low exploration rates, and individuals are more likely to adopt the strategies of their neighbors. When the need for coordination is low, higher exploration rates evolve, leading to lower cultural inertia, and more openness to change. This explains why tight cultures tend to have less deviant behavior among individuals with more norm adherence.

To our knowledge, this is the first paper that predicts the effects of the need for coordination on norm change and cultural inertia, and how it affects an agent's decision of whether to learn from others or to innovate and explore random behaviors. We found our main qualitative findings to be robust to a wide range of parameter values, in both our simulation and theoretical results.

In the future, it would be interesting to study how network structures differ between tight and loose cultures. This would provide the insights necessary to extend our work and study the dynamics of the \emph{rate} at which norm change occurs in different cultures. Further, contrary to prior work on information diffusion (see \cite{shakarian2012review} for a review), our work indicates that the structure of interaction and incentives are vastly different across cultures. Thus, it would be interesting to study culture-sensitive models for information diffusion and propagation, that would be more accurate in predicting cascades and epidemics in different societies around the world.

By studying how socio-structural factors such as the need for coordination affect cultural inertia, this work aims to establish a culturally-sensitive model of norm change. With this model, we identify the conditions that lead to stability or instability in established population norms in different cultural contexts. 
Such knowledge is critical in providing us the ability to identify early markers of impending drastic shifts in populations' norms and thus enable tools providing alerts to potential social uprisings and turmoil.

\section*{Acknowledgments}
This research was supported in part by the U.S. Air Force Grant FA9550-14-1-0020, and was facilitated by the University of Maryland's supercomputing resources (http://www.it.umd.edu/hpcc). The information in this paper does not necessarily reflect the position or policy of the funders, and no official endorsement should be inferred.

\appendix
\setcounter{secnumdepth}{0}
\section{Appendix}

Here we present proofs of the technical results presented in the paper.

\subsection{Proof of Lemma \ref{two_player_nash}}

\begin{proof}
For $(A, A)$ to be a Nash equilibrium of the game $M'$, as defined in Figure \ref{fig:new_payoffs}, the following condition has to hold:
\begin{align}
M'_{AA} \ge M'_{BA} \quad   \Rightarrow   \quad  c \ge \frac{b - a}{b}. \label{twop_a_nash}
\end{align}
Similarly, for $(B, B)$ to be a Nash equilibrium, the required condition is:
\begin{align}
M'_{BB} \ge M'_{AB} \quad   \Rightarrow   \quad  c \ge \frac{a - b}{a}. \label{twop_b_nash}
\end{align}
Consider the following two cases:
\begin{enumerate}
\item $b > a : $ In this case, \eqref{twop_b_nash} is always satisfied. Thus, $(B, B)$ is a NE. If $c$ is large enough such that \eqref{twop_a_nash} is satisfied, then $(A, A)$ is also a NE.

\item $a > b : $ In this case, \eqref{twop_a_nash} is always satisfied. Thus, $(A, A)$ is a NE. If $c$ is large enough such that \eqref{twop_b_nash} is satisfied, then $(B, B)$ is also a NE.
\end{enumerate}

Note that $((q, 1-q), (q, 1-q))$ is a mixed-strategy Nash Equilibrium when the strategy $(q, 1-q)$ makes the agent indifferent to the opponent's strategy, i.e., when:
\begin{align}
qM'_{AA} + (1-q)M'_{BA} = qM'_{AB} + (1-q)M'_{BB}.\label{twop_mixed}
\end{align}
Simplifying this, we get:
$$q = \frac{b - (1-c)a}{c(a + b)}. $$
We know that $0 \le q \le 1$. Thus, this reduces to the following two conditions for $((q, 1-q), (q, 1-q))$ to be a Nash Equilibrium:
$$c \ge \frac{b - a}{b} \quad \text{and} \quad c \ge \frac{a - b}{a}.$$
When $b > a$, $c \ge \frac{a - b}{a}$ is always satisfied. Thus, when $c$ is large enough such that $c \ge \frac{b - a}{b}$, $((q, 1-q), (q, 1-q))$ is a mixed-strategy Nash Equilibrium. Similarly, when $a > b$, $c \ge \frac{b - a}{b}$ is always satisfied, and when $c$ is large enough such that $c \ge \frac{a - b}{a}$, $((q, 1-q), (q, 1-q))$ is a mixed-strategy Nash Equilibrium.
\end{proof}

\subsection{Proof of Lemma \ref{well_mixed_nash}}

\begin{proof}
Consider the cases: $x_A = 0$ (the strategy set where all of the population plays $B$) and $x_A = 1$ (the strategy set where all of the population plays $A$). From Lemma \ref{two_player_nash}, we get the following two cases:

\begin{enumerate}
\item $b > a : $ In this case, \eqref{twop_b_nash} is always satisfied. Thus, $x_A = 0$ is a NE. If $c$ is large enough such that \eqref{twop_a_nash} is satisfied, then $x_A = 1$ is also a NE.

\item $a > b: $ In this case, \eqref{twop_a_nash} is always satisfied. Thus, $x_A = 0$ is a NE. If $c$ is large enough such that \eqref{twop_b_nash} is satisfied, then $x_A = 1$ is also a NE.
\end{enumerate}

Now consider the intermediate case where $x_A = p$ with $0<p<1$. For $x_A = p$ to be a NE, no $A$ agent should have a strictly better payoff if switching to $B$, and vice versa. Thus, the following two conditions need to be simultaneously satisfied:
\begin{align*}
&p M'_{AA} + (1-p) M'_{AB} \ge p M'_{BA} + (1-p) M'_{BB}, \\
\text{and } \quad &p M'_{BA} + (1-p) M'_{BB} \ge p M'_{AA} + (1-p) M'_{AB}.
\end{align*}
Both of these conditions are satisfied only when:
$$p M'_{AA} + (1-p) M'_{AB} = p M'_{BA} + (1-p) M'_{BB}.$$
This simplifies to:
$$p = \frac{b - (1-c) a}{c (a + b)},$$
and similar to Lemma \ref{two_player_nash}, the results follow.
\end{proof}

\subsection{Proof of Theorem \ref{ess_thm}}

\begin{proof}
Let $C$ denote the mixed strategy $(q, 1-q)$. From \eqref{twop_mixed}, we see that for $C$ to be a mixed strategy NE, the following condition needs to hold:
\begin{align}
q = \frac{M'_{BB} - M'_{AB}}{M'_{AA} - M'_{BA} + M'_{BB} - M'_{AB}}, \label{ms_q_cond_ess}
\end{align}
Let $M'_{AC}$ denote the payoff received by the row player when an $A$ player (row player) interacts with a $C$ player (column player). Similarly, we define $M'_{CC}$, $M'_{CA}$, $M'_{BC}$ and $M'_{CB}$. Thus we get:
\begin{align*}
M'_{CC} &:= q^2 M'_{AA} + q(1-q) M'_{AB} + q(1-q)M'_{BA} + (1-q)^2 M'_{BB}, \\
M'_{CA} &:= q M'_{AA} + (1-q) M'_{BA}, \\
M'_{AC} &:= q M'_{AA} + (1-q) M'_{AB}, \\
M'_{CB} &:= q M'_{AB} + (1-q) M'_{BB}, \\
M'_{BC} &:= q M'_{BA} + (1-q) M'_{BB}. 
\end{align*}
Let us derive conditions for which $A$ is an Evolutionarily Stable Strategy (ESS). Thus, we consider the proportion of agents playing $A$ to be close to $1$, i.e., $x_A = 1-\epsilon$, where $0 < \epsilon \ll 1$.

Let $S$ denote the set of strategies other than $A$ that agents can play, i.e., $S \in \{B, C\}$. For $A$ to be an ESS, one of the following conditions need to hold: either
\begin{enumerate*}[label=(\roman*)]
\item $M'_{AA} > M'_{SA}$, or
\item $M'_{AA} = M'_{SA}$ and $M'_{AS} > M'_{SS}$.
\end{enumerate*}

From Lemma \ref{two_player_nash}, we see that $M'_{AA} > M'_{BA}$ simplifies to the condition $c > \frac{b - a}{b}$. Further, $M'_{AA} = M'_{BA}$ simplifies to $c = \frac{b - a}{b}$. 

We also notice that $M'_{AA} > M'_{CA}$ simplifies to:
\begin{align}
M'_{AA} > q M'_{AA} + (1-q) M'_{BA}  \quad \Rightarrow \quad M'_{AA} > M'_{BA}. \nonumber
\end{align}
Now consider the three cases:
\begin{enumerate}
\item $b > a : $ In this case, if $c$ is large enough such that $c > \frac{b - a}{b}$ is satisfied, then $A$ is an ESS.
\item $a > b : $ In this case, $c > \frac{b - a}{b}$ is always satisfied. Thus, $A$ is an ESS.
\item $b > a $ and $c = \frac{b - a}{b} : $ In this case, for $A$ to be an ESS, both the conditions $M'_{AB} > M'_{BB}$ and $M'_{AC} > M'_{CC}$ have to be satisfied.  $M'_{AB} > M'_{BB}$ simplifies to $(1-c) a > b$, which is never satisfied.
$M'_{AC} > M'_{CC}$ simplifies to:
\begin{align}
&q M'_{AA} + (1-q) M'_{AB} > q^2 M'_{AA} + q(1-q) M'_{AB} + q(1-q)M'_{BA} + (1-q)^2 M'_{BB}, \nonumber \\
\Rightarrow \quad & q (1-q) (M'_{AA} - M'_{BA})  > (1 - q)^2 (M'_{BB} - M'_{AB}), \label{mcc_mac}
\end{align}
which is also never satisfied (follows from \eqref{ms_q_cond_ess}). Thus, for this case, $A$ is not an ESS.
\end{enumerate}

We can similarly derive conditions for which $B$ is an Evolutionarily Stable Strategy (ESS).

Now we examine whether $C$ is an ESS. Using \eqref{ms_q_cond_ess}, we can show that the following conditions are satisfied: $M'_{CC} = M'_{AC}$ and $M'_{CC} = M'_{BC}$.  Thus, for $C$ to be an ESS, both $M'_{CA} > M'_{AA}$ and $M'_{CB} > M'_{BB}$ need to be satisfied. These two conditions simplify to the following conditions: $M'_{BA} > M'_{AA}$ and $M'_{AB} > M'_{BB}$, which in turn simplify to the conditions $c < \frac{b - a}{b}$ and $c < \frac{a - b}{a}$. Both of these conditions cannot be simultaneously satisfied. Thus, $C$ is not an ESS.
\end{proof}

{\small
\bibliographystyle{plain}
\bibliography{references}}

\end{document}